\pdfoutput=1
\RequirePackage{ifpdf}
\ifpdf 
\documentclass[pdftex]{sigma}
\else
\documentclass{sigma}
\fi

\numberwithin{equation}{section}

\newtheorem{Theorem}{Theorem}[section]
\newtheorem{Corollary}[Theorem]{Corollary}
\newtheorem{Lemma}[Theorem]{Lemma}
\newtheorem{Proposition}[Theorem]{Proposition}
{ \theoremstyle{definition}
\newtheorem{Definition}[Theorem]{Definition}
\newtheorem{Remark}[Theorem]{Remark} }

\begin{document}
\allowdisplaybreaks

\newcommand{\arXivNumber}{1901.03117}

\renewcommand{\PaperNumber}{067}

\FirstPageHeading

\ShortArticleName{Ergodic Decomposition for Inverse Wishart Measures on Infinite Positive-Definite Matrices}

\ArticleName{Ergodic Decomposition for Inverse Wishart Measures\\ on Infinite Positive-Definite Matrices}

\Author{Theodoros ASSIOTIS}

\AuthorNameForHeading{T.~Assiotis}

\Address{Mathematical Institute, University of Oxford, Oxford, OX2 6GG, UK}
\Email{\href{mailto:theo.assiotis@maths.ox.ac.uk}{theo.assiotis@maths.ox.ac.uk}}
\URLaddress{\url{https://sites.google.com/view/theoassiotis/home}}

\ArticleDates{Received April 08, 2019, in final form September 04, 2019; Published online September 11, 2019}

\Abstract{The ergodic unitarily invariant measures on the space of infinite Hermitian matrices have been classified by Pickrell and Olshanski--Vershik. The much-studied complex inverse Wishart measures form a projective family, thus giving rise to a unitarily inva\-riant measure on infinite positive-definite matrices. In this paper we completely solve the corresponding problem of ergodic decomposition for this measure.}

\Keywords{infinite random matrices; ergodic measures; inverse Wishart measures; ortho\-gonal polynomials}

\Classification{60B15; 60G55}

\section{Introduction}\label{section1}

\subsection{Informal introduction and historical overview}\label{section1.1}
Since their introduction in the 1920's \cite{WishartOriginal} the Wishart measures have been ubiquitous in mathematics, physics and statistics. They appear in diverse fields, from statistical analysis \cite{StatisticalMatrixAnalysis, Johnstone} to stochastic processes \cite{Wishart, RiderValko} and free probability \cite{FreeCapitaine, FreeSpeicher}. More recently, there has been renewed interest in connection to their applications in quantum transport \cite{QuantumTransport1,QuantumTransport2} and their enumerative properties, in particular relations to Hurwitz numbers \cite{Hurwitz}.

This paper stems from the following remarkable property, and its consequences, of these measures: the inverse Wishart measures $\big\{\mathsf{M}^{(\nu),N}\big\}_{N\ge 1}$ ($\nu$ is a real parameter greater than~$-1$) defined in (\ref{MeasureDefinition}) on (positive-definite) Hermitian matrices form a projective family under the so-called corners maps given in (\ref{Corners}) below. The goal of this paper is to study the corresponding unitarily invariant, namely invariant under conjugation by unitary matrices, measure $\mathsf{M}^{(\nu)}$ on infinite (positive-definite) Hermitian matrices and describe explicitly how it decomposes into ergodic components.

These ergodic measures for the action by conjugation of the inductive limit of unitary groups on infinite Hermitian matrices have been classified in classical works of Pickrell~\cite{Pickrell} and Olshanski and Vershik~\cite{OlshanskiVershik}. They are parametrized by the infinite dimensional space $\Omega$ defined in (\ref{Omegadefinition}) below and depend on a set of parameters $\big(\{\alpha^+\},\{\alpha^-\},\gamma_1,\gamma_2\big)\subset \mathbb{R}_+^{\infty}\times \mathbb{R}_+^{\infty}\times\mathbb{R}\times\mathbb{R}_+$. As we will see in Section~\ref{section3}, the $\alpha$ parameters are asymptotic eigenvalues and $\gamma_1$ and $\gamma_2$ are related to the asymptotic trace and asymptotic sum of squares respectively; the parameter $\gamma_2$ is also called the `Gaussian component'. The study of $\gamma_1$ especially and also~$\gamma_2$ is a much more difficult task compared to describing the $\alpha$'s.

The problem of ergodic decomposition of certain distinguished unitarily invariant measures on infinite Hermitian matrices was initiated by Borodin and Olshanski in \cite{BorodinOlshanskiErgodic}. They considered the Hua--Pickrell measures depending on a complex parameter $\mathsf{s}$. These measures were first studied by Hua in his classical book \cite{Hua} and implicit in his calculations is consistency under the corners maps, see also Neretin's generalization~\cite{Neretin}. Borodin and Olshanski described the $\alpha$ parameters and proved that for $\mathsf{s}=0$, $\gamma_2=0$. The determination of $\gamma_1$ and $\gamma_2$ was left open for many years until recently in a breakthrough work~\cite{Qiu} Qiu proved that $\gamma_2=0$ for general parameter $\mathsf{s}$ and completely described $\gamma_1$ for real $\mathsf{s}$ (see Remark~\ref{EstimatesHuaPickrell} for more on this restriction). In the case of the inverse Wishart measures $\mathsf{M}^{(\nu)}$ we are able to completely describe all the parameters for all $\nu>-1$. This is achieved in Theorem~\ref{MainTheorem}, the main result of this paper.

A closely related problem is the ergodic decomposition of the so-called Pickrell measures \cite{Pickrellmeasure} (depending on a real parameter $\mathfrak{s}$) on infinite square complex matrices. The ergodic unitarily invariant (by multiplication to the left and to the right) measures on infinite square complex matrices have also been classified. These are parametrized by a different infinite dimensional space that is a subset of $\mathbb{R}_+^{\infty} \times \mathbb{R}_+$ (there is no analogue of $\gamma_1$). The explicit description of the ergodic decomposition has been settled in a series of papers by Bufetov \cite{BufetovI,BufetovII,BufetovIII} (see also~\cite{BufetovQiu}), which have been very influential for us. We should mention that the papers of Bufetov \cite{BufetovI, BufetovII,BufetovIII} and Qiu~\cite{Qiu} also study the infinite case of the problem of ergodic decomposition, namely when the corresponding matrix measures no longer have finite mass. Since this requires quite different techniques we will not consider it in this work.

Finally, before closing this informal introduction we remark that a key role in all these papers \cite{BorodinOlshanskiErgodic,BufetovI,BufetovII,BufetovIII,BufetovQiu,Qiu} is played by orthogonal polynomials. In the case of the Hua--Pickrell measures these are the pseudo-Jacobi polynomials and in the case of the Pickrell measures these are the Jacobi. The analogous role in this paper is played by the Bessel \cite{BesselPolynomials} and also Laguerre polynomials.

In the next subsections we give the necessary background to make the informal discussion above precise and state our results rigorously.

\subsection{Ergodic unitarily invariant measures on infinite Hermitian matrices}\label{section1.2}

Let $\mathbb{U}(N)$ and $H(N)$ be the group of $N \times N$ unitary matrices and the space of $N \times N$ Hermitian matrices respectively. Let $\mathbb{U}(\infty)$ be the inductive limit of unitary groups
\begin{gather*}
\mathbb{U}(\infty)=\mathop{\underset{\rightarrow}{\lim}}\mathbb{U}(N),
\end{gather*}
 under the natural embeddings. In more explicit terms an element of $\mathbb{U}(\infty)$ is an infinite block matrix whose top corner is an $N\times N$ unitary matrix for some finite $N$ and the other block is the (infinite) identity matrix.

Consider the so called corners maps $\pi^N_{N-1}\colon H(N)\to H(N-1)$ defined by
\begin{align}\label{Corners}
\pi_{N-1}^N\big[(h_{ij})_{i,j=1}^N\big]= (h_{ij} )_{i,j=1}^{N-1}.
\end{align}
Let $H$ be the space of infinite Hermitian matrices defined as the projective limit
\begin{gather*}
H=\mathop{\underset{\leftarrow}{\lim}}H(N),
\end{gather*}
under the corners maps. Moreover, let $H_+(N)\subset H(N)$ denote the space of $N \times N$ positive-definite matrices; namely matrices with positive eigenvalues. By Cauchy's interlacing theorem we get that $\pi^N_{N-1}\colon H_+(N)\to H_+(N-1)$. Thus, we can also correctly define the projective limit $H_+=\underset{\leftarrow}{\lim}H_+(N) \subset H$.

Now, $\mathbb{U}(\infty)$ acts on $H$ by conjugation: for each $u \in \mathbb{U}(\infty)$ we have a map $T_u\colon H \to H$ given by $T_u(h)=u^*hu$. It is a classical theorem of Pickrell \cite{Pickrell} and also Olshanski and Vershik \cite{OlshanskiVershik} that ergodic measures on $H$ for this action (namely ones such that all $\mathbb{U}(\infty)$-invariant subsets have mass 0 or 1) are parametrized by the infinite dimensional space $\Omega \subset \mathbb{R}^{2\infty+2}$:
\begin{gather}
\Omega =\Big\{\omega=(\alpha^+,\alpha^-,\gamma_1,\delta)\in \mathbb{R}^{2\infty+2}=\mathbb{R}^\infty \times \mathbb{R}^\infty \times \mathbb{R} \times \mathbb{R}\,|\nonumber\\
\hphantom{\Omega =\Big\{}{}\alpha^+ =(\alpha_1^+\ge \alpha_2^+\ge \cdots \ge 0) ; \, \alpha^-=(\alpha_1^-\ge \alpha_2^-\ge \cdots \ge 0);\nonumber\\
\hphantom{\Omega =\Big\{}{} \gamma_1 \in \mathbb{R} ;\, \delta \ge 0 ;\,\sum_{}^{}(\alpha_i^+)^2 + \sum_{}^{}(\alpha_i^-)^2\le \delta \Big\},\nonumber\\
 \gamma_2=\delta-\sum_{}^{}\big(\alpha_i^+\big)^2 - \sum_{}^{}\big(\alpha_i^-\big)^2.\label{Omegadefinition}
\end{gather}
Observe that $\Omega$ is a locally compact space. We then have the following classification theorem, see \cite{OlshanskiVershik,Pickrell}, also \cite{Defosseux}.

\begin{Theorem}[Pickrell, Olshanski--Vershik] There exists a parametrization of ergodic $\mathbb{U}(\infty)$-invariant probability measures on the space $H$ by the points of the space $\Omega$ described as follows. Given $\omega \in \Omega$ the characteristic function of the ergodic measure $M_{\omega}$ is given by
\begin{gather*}
\int_{X \in H}^{}{\rm e}^{{\rm i}\operatorname{Tr}(\operatorname{diag}(r_1,\dots,r_n,0,0,\dots)X)}M_{\omega}({\rm d}X)=\prod_{j=1}^{n}F_{\omega}(r_j),
\end{gather*}
where
\begin{gather*}
F_{\omega}(x)={\rm e}^{{\rm i}\gamma_1x-\frac{\gamma_2}{2}x^ 2}\prod_{k=1}^{\infty}\frac{{\rm e}^{-{\rm i}\alpha_k^+x}}{1-{\rm i}\alpha_k^+x}\prod_{k=1}^{\infty}\frac{{\rm e}^{{\rm i}\alpha_k^-x}}{1+{\rm i}\alpha_k^-x} .
\end{gather*}
\end{Theorem}
Note that the characteristic function $F_{\omega}$ is well defined for all $\omega \in \Omega$ by the fact that the sum of squares of the $\alpha$'s is finite. Also, observe that parameter $\gamma_2$ corresponds to an infinite random Hermitian matrix with independent (subject to the Hermitian constraint) complex Gaussian entries with mean 0 and variance $\gamma_2$ on the diagonal.

\subsection[The inverse Wishart measures $\mathsf{M}^{(\nu)}$ on infinite positive-definite matrices]{The inverse Wishart measures $\boldsymbol{\mathsf{M}^{(\nu)}}$ on infinite positive-definite matrices}\label{section1.3}

For $\nu>-1$, consider the complex Wishart (or Laguerre ensemble) probability measure on $N \times N$ Hermitian matrices, supported on $H_+(N)$:
\begin{gather*}
\mathcal{M}^{(\nu),N}({\rm d}Y)=\widetilde{\operatorname{const}_{\nu, N}}\det(Y)^{\nu} {\rm e}^{-\operatorname{Tr}Y}\mathbf{1}_{\{Y \in H_+(N)\}}{\rm d}Y\,
\end{gather*}
where throughout the paper for a Hermitian matrix $Y$, ${\rm d}Y$ denotes Lebesgue measure on $H(N)$ (we suppress dependence on~$N$):
\begin{gather*}
{\rm d}Y=\prod_{j=1}^{N}{\rm d}Y_{jj} \prod_{1 \le j <k \le N}^{}{\rm d} \operatorname{Re} Y_{jk}{\rm d} \operatorname{Im} Y_{jk}.
\end{gather*}
The restriction $\nu>-1$ is so that the normalization constant $\widetilde{\operatorname{const}_{\nu, N}}$ is finite.

Under the change of variables $Y=2X^{-1}$ we obtain the central object of study in this paper, the inverse Wishart probability measures on $H_+(N)$
\begin{align}\label{MeasureDefinition}
\mathsf{M}^{(\nu),N}({\rm d}X)= \operatorname{const}_{\nu, N} \det(X)^{-\nu-2N} {\rm e}^{-2\operatorname{Tr}X^{-1}}\mathbf{1}_{ \{X \in H_+(N) \}}{\rm d}X.
\end{align}
Observe that, for all $N \in \mathbb{N}$ the measure $\mathsf{M}^{(\nu),N}$ is unitarily invariant, namely invariant under the action of $\mathbb{U}(N)$ by conjugation. Then, we have the following consistency result:

\begin{Proposition}\label{ConsistencyIntro} For $\nu >-1$ the measures $\big\{\mathsf{M}^{(\nu),N}\big\}_{N\ge 1}$ form a projective family
\begin{gather*}
\big(\pi_{N-1}^N\big)_*\mathsf{M}^{(\nu),N}=\mathsf{M}^{(\nu),N-1}.
\end{gather*}
\end{Proposition}

Thus, by Kolmogorov's theorem we obtain a unitarily invariant measure $\mathsf{M}^{(\nu)}$ on $H$ that is supported on $H_{+}$.

Proposition \ref{ConsistencyIntro} is proven in Section~\ref{section2} as Proposition \ref{consistencymatrices}. A key role in the proof is played by the Bessel orthogonal polynomials \cite{BesselPolynomials}.

\begin{Remark}Although we have not been able to locate Proposition~\ref{ConsistencyIntro} anywhere in the lite\-ra\-ture, in an equivalent form a close variant of it, for the analogous measure on real symmetric positive-definite matrices, appears to be folklore in the statistical literature, see for example \cite[Chapter~3]{StatisticalMatrixAnalysis}. The argument there makes use of the technique of Schur complementation, the formula for the determinant of a block matrix and the special form of the density of $\mathcal{M}^{(\nu),N}$. The proof presented in Section~\ref{section2} is completely different and makes a novel use of the connection to orthogonal polynomials.
\end{Remark}

\subsection[Description of ergodic decomposition of $\mathsf{M}^{(\nu)}$]{Description of ergodic decomposition of $\boldsymbol{\mathsf{M}^{(\nu)}}$}\label{section1.4}

Consider the measure $\mathsf{M}^{(\nu)}$ defined above for $\nu >-1$. It is a result of Borodin and Olshanski, see Proposition~4.4 in~\cite{BorodinOlshanskiErgodic}, that there exists a unique probability measure $\mathfrak{m}^{(\nu)}$ on $\Omega$ such that
\begin{gather*}
\mathsf{M}^{(\nu)}=\int_{\Omega}^{}M_{\omega}\mathfrak{m}^{(\nu)}({\rm d}\omega).
\end{gather*}
Here, the equality is interpreted as integrated against a test function, see Section~\ref{section3}.
The main result of this paper is the explicit description of~$\mathfrak{m}^{(\nu)}$. To proceed and state it precisely we need some more definitions and background on determinantal measures.

\begin{Definition} Define the subset $\Omega_0^+ \subset \Omega$ such that $\omega=\left(\alpha^+,\alpha^-, \gamma_1, \delta\right) \in \Omega_0^+$ iff
\begin{gather*}
\alpha_i^{-}(\omega) \equiv 0, \qquad \alpha_j^{+}(\omega) \neq 0, \qquad \text{for all} \quad i,j \in \mathbb{N};\\
\gamma_2(\omega) =\delta(\omega)-\sum_{i}^{}\big(\alpha_i^+(\omega)\big)^2-\sum_{i}^{}\big(\alpha_i^-(\omega)\big)^2=0;\\
\gamma_1(\omega) = \sum_{i }^{}\alpha_i^{+}(\omega)<\infty.
\end{gather*}
\end{Definition}
In particular, on $\Omega_0^+$ there is no `Gaussian component' namely $\gamma_2(\omega)\equiv 0$ and the para\-me\-ter~$\gamma_1(\omega)$ is completely determined by the $\alpha_i^+(\omega)$.

\paragraph{Determinantal measures}
Let $\mathcal{X}$ be a locally compact Polish space equipped with a $\sigma$-finite reference measure $\mu$. Let $\mathsf{Conf}(\mathcal{X})$ be the space of point configurations over $\mathcal{X}$. Points in a point configuration $\mathfrak{X} \in \mathsf{Conf}(\mathcal{X})$ will be called particles. We can embed $\mathsf{Conf}(\mathcal{X})$ in the space of finite measures on $\mathcal{X}$ by $\mathfrak{X}\mapsto \sum\limits_{x \in \mathfrak{X}}^{}\delta_x$ and with the induced topology $\mathsf{Conf}(\mathcal{X})$ is a Polish space.

A Borel probability measure $\mathbb{P}$ on $\mathsf{Conf}(\mathcal{X})$ is called a determinantal point process or measure, see \cite{Soshnikov}, with (Hermitian) kernel $\mathsf{K}\colon \mathcal{X} \times \mathcal{X} \to \mathbb{C}$ if for any $n\in \mathbb{N}$ and function $\Phi\in C_c\big(\mathcal{X}^n\big)$, continuous and of compact support, we have
\begin{gather}
\int_{\mathsf{Conf}(\mathcal{X})}^{} \sum_{x_1, \dots,x_n \in \mathfrak{X}}\Phi(x_1,\dots,x_n)\mathbb{P}({\rm d}\mathfrak{X})\nonumber\\
\qquad{} =\int_{\mathcal{X}^n}^{}\Phi(x_1,\dots,x_n)\det (\mathsf{K}(x_i,x_j) )_{i,j=1}^n{\rm d}\mu(x_1)\cdots {\rm d}\mu(x_n),\label{determinantaldefinition}
\end{gather}
where the sum is taken over ordered $n$-tuples of particles with pairwise distinct labels. The measure $\mathbb{P}$ satisfying (\ref{determinantaldefinition}) completely determines the pair $(\mathsf{K},\mu)$ and dropping dependence on $\mu$ (usually fixed), we denote it by $\mathbb{P}_{\mathsf{K}}$.

Throughout this paper the reference measure will be Lebesgue measure. We now define a~distinguished determinantal measure on $(0,\infty)$.

\begin{Definition}Define the Bessel kernel $\mathbb{J}_{\nu}$ by
\begin{gather*}
\mathbb{J}_{\nu}(x,y)=\frac{\sqrt{x}J_{\nu+1}(\sqrt{x})J_{\nu}(\sqrt{y})-\sqrt{y}J_{\nu+1}(\sqrt{y})J_{\nu}(\sqrt{x})}{2(x-y)},
\end{gather*}
where $J_{\nu}$ is the Bessel function of order $\nu$. This kernel gives rise to the Bessel determinantal point process $\mathbb{P}_{\mathbb{J}_{\nu}}$ on $(0,\infty)$ of Tracy and Widom \cite{TracyWidom}.
\end{Definition}

We are ready to give the explicit description of the ergodic decomposition of the inverse Wishart measures $\mathsf{M}^{(\nu)}$.

\begin{Theorem}\label{MainTheorem}
Let $\nu> -1$. The spectral measure $\mathfrak{m}^{(\nu)}$ associated to $\mathsf{M}^{(\nu)}$ is concentrated on~$\Omega_0^+$
\begin{gather*}
\mathfrak{m}^{(\nu)}\big(\Omega_0^+\big)=1.
\end{gather*}
Moreover, the law of the parameters $\alpha^+=\big(\alpha_1^+\ge \alpha_2^+ \ge \alpha_3^+ \ge \cdots \ge 0 \big)$ under $\mathfrak{m}^{(\nu)}$ viewed as a~point configuration on $(0,\infty)$ is given by the determinantal point process $\mathbb{P}_{K^{\nu}_{\infty}}$ with correlation kernel $K_{\infty}^{\nu}(x,y)$ defined by
\begin{gather*}
K_{\infty}^{\nu}(x,y)=\frac{8}{xy} \mathbb{J}_{\nu}\left(\frac{8}{x},\frac{8}{y}\right).
\end{gather*}
\end{Theorem}

Theorem \ref{MainTheorem} above will be proven by an approximation procedure, the so-called `ergodic method' of Olshanski and Vershik \cite{OlshanskiVershik}, that is explained in Section~\ref{section3}. For this asymptotic analysis we will exploit the relation with the Laguerre ensemble.

\subsection{Organisation of the paper}\label{section1.5}

In Section~\ref{section2}, we prove consistency for the inverse Wishart measures. This is proved at the level of eigenvalue measures first and then lifted to matrices. The key ingredient is the backward shift equation for Bessel polynomials. In Section~\ref{section3}, we recall in detail the approximation procedure of Vershik, Borodin and Olshanski. In Section~\ref{section4}, we study the $\alpha$ parameters. We obtain the parameters~$\alpha^+$ as limits of the Bessel orthogonal polynomial ensemble, which is a determinantal point process. In Section~\ref{section5}, we obtain certain estimates for the kernel of the Bessel orthogonal polynomial ensemble. These are essential for the study of $\gamma_1$ and $\gamma_2$. In Section~\ref{section6}, we prove that $\gamma_2 \equiv 0$. In Section~\ref{section7}, we treat the $\gamma_1$ parameter. A key role is played by positivity along with the estimates from Section~\ref{section5}. Finally, in Section~\ref{section8} we simply put everything together to conclude the proof of Theorem~\ref{MainTheorem}.

\section[Consistency of $\mathsf{M}^{(\nu),N}$]{Consistency of $\boldsymbol{\mathsf{M}^{(\nu),N}}$}\label{section2}

Let $\nu>-1$. Recall that, we are interested in the normalized (probability) measure
\begin{gather*}
\mathsf{M}^{(\nu),N}({\rm d}X)= \operatorname{const}_{\nu, N} \det(X)^{-\nu-2N} {\rm e}^{-2\operatorname{Tr}X^{-1}}\mathbf{1}_{\{X \in H_+(N)\}}{\rm d}X,
\end{gather*}
where the normalizing constant $\operatorname{const}_{\nu, N}$ will be given explicitly below.

By Weyl's integration formula we get that the induced probability measure on eigenvalues in the Weyl chamber $W^N$ (we write $W^N_+$ if all coordinates are positive)
\begin{gather*}
W^N=\big\{x=(x_1,\dots,x_N) \in \mathbb{R}^N\colon x_1 \ge x_2 \ge \cdots \ge x_N \big\}
\end{gather*}
is given by
\begin{gather*}
\mu_N^{\nu}({\rm d}x)=\overline{\operatorname{const}}_{\nu, N}\prod_{i=1}^{N}w_N^{\nu}(x_i)\Delta_N(x)^2\mathbf{1}_{ \{x \in W_+^N \}}{\rm d}x_1\cdots {\rm d}x_N
\end{gather*}
with the Bessel weight $w_N^{\nu}(\cdot)$ on $(0,\infty)$
\begin{gather*}
w_N^{\nu}(x)=x^{-\nu-2N}{\rm e}^{-\frac{2}{x}}
\end{gather*}
and we will write
\begin{gather*}
\Delta_N(x)=\prod_{1 \le i<j \le N}^{}(x_i-x_j)
\end{gather*}
for the Vandermonde determinant.

Moreover, the constants $\operatorname{const}_{\nu, N}$ and $\overline{\operatorname{const}}_{\nu, N}$ are related by
\begin{gather*}
\operatorname{const}_{\nu, N}=\overline{\operatorname{const}}_{\nu, N}\frac{1}{(2\pi)^\frac{N(N-1)}{2}}\prod_{k=1}^{N}k!.
\end{gather*}

It will be convenient to introduce the following notation: we define the map $\mathsf{eval}_N\colon H(N) \to W^N$ by $\mathsf{eval}_N(H)=x=(x_1\ge x_2 \ge \cdots \ge x_N )$ the vector of eigenvalues of $H$ in weakly decreasing order. In this notation we have
\begin{gather*}
 (\mathsf{eval}_N )_*\mathsf{M}^{(\nu),N}=\mu_N^{\nu}.
\end{gather*}

{\bf Bessel polynomials.} The orthogonal polynomials with respect to $w_N^{\nu}(x)$, that are called the Bessel polynomials \cite{BesselPolynomials}, will be important for us. The reference for all the facts stated below is \cite[Chapter~9.13, pp.~244--247]{HypergeometricOrthogonalPolynomials}.

 For $\nu>-1$, there exist $p_0(\cdot;\nu,N),\dots,p_{N-1}(\cdot;\nu,N)$ \textit{monic} orthogonal polynomials of degree $0,\dots, N-1$ with respect to $w_N^{\nu}$. These can be expressed in terms of hypergeometric functions but we shall not need any explicit expression here. Their norms are given by, for $n=0, \dots, N-1$,
\begin{align*}
\|p_n(\cdot;\nu,N)\|_2^2&= \int_{0}^{\infty}w_N^{\nu}(x)p_n^2(x;\nu,N){\rm d}x\\
&=-\frac{2^{2n-\nu-2N+1}}{(n-\nu-2N+1)_n^2(2n-\nu-2N+1)}\Gamma(-n+\nu+2N)n!.
\end{align*}
Here, $(a)_k=\prod\limits_{i=1}^{k}(a+i-1), (a)_0=1$ is the Pochhammer symbol.

Also, we have the following key relation, called the backward shift equation, that relates orthogonal polynomials with respect to $w_N^{\nu}$ to orthogonal polynomials with respect to~$w_{N+1}^\nu$
\begin{align*}
\frac{{\rm d}}{{\rm d}x}\big[w_N^{\nu}(x)p_n(x;\nu,N)\big]&=\frac{(n+1-\nu-2N+1)_{n+1}}{(n-\nu-2N+1)_n}w_{N+1}^{\nu}(x)p_{n+1}(x;\nu, N+1)\nonumber\\
&=c(n,\nu,N)w_{N+1}^{\nu}(x)p_{n+1}(x;\nu, N+1).
\end{align*}

Furthermore, by writing the Vandermonde determinant in terms of the monic orthogonal polynomials $p_n(\cdot; \nu, N)$ a standard calculation gives that{\samepage
\begin{gather*}
\overline{\operatorname{const}}_{\nu, N}=\prod_{n=1}^{N}\frac{1}{\|p_{n-1}(\cdot;\nu,N)\|_2^2}=\prod_{n=1}^{N}-\frac{(n-\nu-2N+1)_n^2(2n-\nu-2N+1)}{2^{2n-\nu-2N+1}\Gamma(-n+\nu+2N)n!}
\end{gather*}
and thus we also obtain an explicit expression for $\operatorname{const}_{\nu, N}$.}

Now, we introduce the following Markov kernel $\Lambda_N^{N+1}$ from $W^{N+1}$ to $W^N$, defined for \mbox{$x \in \mathring{W}^{N+1}$} (the interior of $W^{N+1}$) as
\begin{gather*}
\Lambda_N^{N+1}(x,{\rm d}y)=\frac{N!\Delta_N(y)}{\Delta_{N+1}(x)}\textbf{1}(y\prec x){\rm d}y,
\end{gather*}
where $y\prec x$ denotes interlacing
\begin{gather*}
x_1 \ge y_1 \ge x_2 \ge \cdots \ge y_N \ge x_{N+1}.
\end{gather*}
This Markov kernel has a random matrix interpretation, due to Baryshnikov~\cite{Baryshnikov} (the required computation is already implicit in the classical book of Gelfand and Naimark \cite{GelfandNaimark}), which easily extends to any $x\in W^{N+1}$ (even when the coordinates coincide)
\begin{gather*}
\Lambda_N^{N+1}(x,\cdot)=\mathsf{Law} \big[\mathsf{eval}_N \big(\pi_N^{N+1}\big(U^*\mathsf{diag}(x_1,\dots, x_{N+1})U\big)\big)\big],
\end{gather*}
where $U$ is a random (Haar distributed) unitary matrix from $\mathbb{U}(N+1)$.

We will first prove consistency at the level of eigenvalues:

\begin{Proposition}\label{eigenvalueconsistency}
For $\nu>-1$
\begin{gather*}
\mu_{N+1}^{\nu} \Lambda_N^{N+1}=\mu_N^{\nu}, \qquad \forall\, N\ge 1.
\end{gather*}
\end{Proposition}

\begin{proof}Since both sides are probability measures, it suffices to prove that they are equal up to multiplicative constant (we will use the notation $\propto$ for this). Since $\mu_{N+1}^{\nu}$ is supported on $\mathring{W}^{N+1}$ we can write (note that since $y\prec x$ we have $\big\{x \in W^{N+1}_+\big\} \Rightarrow \big\{y \in W^{N}_+\big\}$)
\begin{gather*}
\big[\mu_{N+1}^{\nu}\Lambda_N^{N+1}\big]({\rm d}y)=N! \Delta_N(y)\mathbf{1}_{\{y \in W_+^N \}}{\rm d}y \int_{y\prec x \in W^{N+1}_+}^{}\prod_{i=1}^{N+1}w_{N+1}^{\nu}(x_i) \Delta_{N+1}(x){\rm d}x.
\end{gather*}
Note that we can write
\begin{gather*}
\Delta_{N+1}(x)=\det [p_{i-1} (x_{N+2-j};\nu,N+1 ) ]^{N+1}_{i,j=1}.
\end{gather*}
We thus need to show
\begin{gather}
\int_{y\prec x \in W^{N+1}_+}^{}\prod_{i=1}^{N+1}w_{N+1}^{\nu}(x_i) \det [p_{i-1} (x_{N+2-j};\nu,N+1 ) ]^{N+1}_{i,j=1}{\rm d}x\nonumber\\
\qquad {}\propto \prod_{i=1}^{N}w_N^{\nu}(y_i)\det [p_{i-1} (y_{N+1-j};\nu,N ) ]^{N}_{i,j=1}
 =\prod_{i=1}^{N}w_N^{\nu}(y_i)\Delta_N(y) .\label{proportionalmeasures}
\end{gather}
By multilinearity of the determinant we obtain that the l.h.s.\ of (\ref{proportionalmeasures}) is
\begin{gather*}
\det\left[\int_{y_{N+2-j}}^{y_{N+1-j}}w_{N+1}^{\nu}(z)p_{i-1} (z;\nu,N+1 ){\rm d}z\right]^{N+1}_{i,j=1},
\end{gather*}
where $y_0=\infty, y_{N+1}=0$.

Now, by the backward shift equation we can perform the integral inside the determinant and obtain for $i\ge 2$:
\begin{gather*}
\int_{y_{N+2-j}}^{y_{N+1-j}}w_{N+1}^{\nu}(z)p_{i-1} (z;\nu,N+1 ){\rm d}z=\frac{1}{c(i-2,\nu,N)}\big[w_{N}^{\nu}(z)p_{i-2}(z;\nu, N)\big]^{y_{N+1-j}}_{y_{N+2-j}}.
\end{gather*}
We note that when evaluated at $y_0=\infty$ and $y_{N+1}=0$ the terms above vanish. Hence, we get that the l.h.s.\ of (\ref{proportionalmeasures}) is proportional to (note that the entry with index $(2,2)$ is the difference $w_N^{\nu}(y_{N-1})p_0(y_{N-1};\nu, N)-w_N^{\nu}(y_N)p_0(y_N;\nu, N)$ and not just $w_N^{\nu}(y_{N-1})p_0(y_{N-1};\nu, N)$, similarly for the entry with index~$(N+1,2)$)
\begin{gather*}
\det\begin{bmatrix}
\displaystyle \int_{0}^{y_N}\!w_{N+1}^{\nu}(z){\rm d}z &\displaystyle \int_{y_N}^{y_{N-1}}\!w_{N+1}^{\nu}(z){\rm d}z &\cdots &\displaystyle \int_{y_1}^{\infty}\!w_{N+1}^{\nu}(z){\rm d}z\vspace{1mm}\\
 w_N^{\nu}(y_N)p_0(y_N;\nu, N) & \begin{matrix} w_N^{\nu}(y_{N-1})p_0(y_{N-1};\nu, N)\\ -w_N^{\nu}(y_N)p_0(y_N;\nu, N)\end{matrix} &\cdots & -w_N^{\nu}(y_1)p_0(y_1;\nu, N)\\
\vdots &\vdots &\vdots &\vdots\\
 w_N^{\nu}(y_N)p_{N-1}(y_N;\nu, N)\!\! & \begin{matrix} w_N^{\nu}(y_{N-1})p_{N-1}(y_{N-1};\nu, N)\\ -w_N^{\nu}(y_N)p_{N-1}(y_N;\nu, N)\end{matrix}\!\! &\cdots & -w_N^{\nu}(y_1)p_{N-1}(y_1;\nu, N)
\end{bmatrix}.
\end{gather*}
Successively adding column $j$ to column $j+1$ we obtain that this is equal to
\begin{gather*}
\det\begin{bmatrix}
\displaystyle \int_{0}^{y_N}w_{N+1}^{\nu}(z){\rm d}z & \cdots\ &\cdots \ &\displaystyle \int_{0}^{\infty}w_{N+1}^{\nu}(z){\rm d}z\vspace{1mm}\\
 w_N^{\nu}(y_N)p_0(y_N;\nu, N) &\cdots & w_N^{\nu}(y_1)p_{0}(y_1;\nu, N) & 0\\
\vdots &\vdots &\vdots &\vdots\\
 w_N^{\nu}(y_N)p_{N-1}(y_N;\nu, N) & \cdots & w_N^{\nu}(y_1)p_{N-1}(y_1;\nu, N) & 0
\end{bmatrix}
\end{gather*}
from which (\ref{proportionalmeasures}) immediately follows.
\end{proof}

\begin{Remark}
As pointed out to me by Grigori Olshanski a similar idea of using the backwards shift equation to prove consistency appears in the study of the q-zw measures on the quantized Gelfand--Tsetlin graph \cite{GorinOlshanski}. The corresponding orthogonal polynomials are the pseudo big $q$-Jacobi, see \cite{GorinOlshanski} and the references therein.
\end{Remark}

We will now, making use of Baryshnikov's result \cite{Baryshnikov}, prove consistency for the matrix measures.

\begin{Proposition}\label{consistencymatrices}For $\nu>-1$, the inverse Wishart measures form a projective family
\begin{gather*}
\big(\pi_{N}^{N+1}\big)_*\mathsf{M}^{(\nu),N+1}=\mathsf{M}^{(\nu),N}, \qquad \forall\, N\ge 1.
\end{gather*}
Thus, by Kolmogorov's theorem there exists a unique unitarily invariant measure $\mathsf{M}^{(\nu)}$ on $H$ that is supported on~$H_{+}$.
\end{Proposition}

\begin{proof}First, observe that if a measure $\mathfrak{N}$ is unitarily invariant on $H(N+1)$ then $\big(\pi_{N}^{N+1}\big)_*\mathfrak{N}$ is unitarily invariant on $H(N)$. Thus, it suffices to show that
\begin{gather*}
 (\mathsf{eval}_N )_*\big(\pi_N^{N+1}\big)_*\mathsf{M}^{(\nu),N+1}= (\mathsf{eval}_N )_*\mathsf{M}^{(\nu),N}.
\end{gather*}
On the other hand from Proposition~\ref{eigenvalueconsistency} we have
\begin{gather*}
\left(\mathsf{eval}_{N+1}\right)_*\mathsf{M}^{(\nu),N+1}\circ \Lambda_N^{N+1}=\left(\mathsf{eval}_N\right)_*\mathsf{M}^{(\nu),N}.
\end{gather*}
Hence, we need to show
\begin{gather*}
 (\mathsf{eval}_N )_*\big(\pi_N^{N+1}\big)_*\mathsf{M}^{(\nu),N+1}= (\mathsf{eval}_{N+1} )_*\mathsf{M}^{(\nu),N+1}\circ \Lambda_N^{N+1}.
\end{gather*}
Now, let $\lambda=(\lambda_1\ge \lambda_2 \ge \cdots \ge \lambda_{N+1})$ be fixed and consider the orbital measure on $H(N+1)${\samepage
\begin{gather*}
\mathsf{n}_{\lambda}=\mathsf{Law} [U^*\mathsf{diag}(\lambda_1,\dots, \lambda_{N+1})U ],
\end{gather*}
where $U$ is Haar distributed in $\mathbb{U}(N+1)$.}

Then, since $ (\mathsf{eval}_{N+1} )_*\mathsf{n}_{\lambda}=\delta_{\lambda}$, Baryshnikov's result~\cite{Baryshnikov} can be written as
\begin{gather*}
\big(\mathsf{eval}_N\circ\pi_N^{N+1}\big)_*\mathsf{n}_{\lambda}= (\mathsf{eval}_{N+1} )_*\mathsf{n}_{\lambda}\circ \Lambda_N^{N+1}.
\end{gather*}
By Weyl's integration formula we have
\begin{gather*}
\mathsf{M}^{(\nu),N+1}({\rm d}X)=\int\mu_{N+1}^{\nu}({\rm d}\lambda)\mathsf{n}_{\lambda}({\rm d}X).
\end{gather*}
So we obtain
\begin{gather*}
\int \big(\mathsf{eval}_N\circ\pi_N^{N+1}\big)_*\mathsf{n}_{\lambda}(\cdot) \mu_{N+1}^{\nu}({\rm d}\lambda)= \int (\mathsf{eval}_{N+1} )_*\mathsf{n}_{\lambda}\circ \Lambda_N^{N+1}(\cdot) \mu_{N+1}^{\nu}({\rm d}\lambda).
\end{gather*}
By linearity of the operations involved
\begin{gather*}
\big(\mathsf{eval}_N\circ\pi_N^{N+1}\big)_*\int \mathsf{n}_{\lambda}(\cdot) \mu_{N+1}^{\nu}({\rm d}\lambda)= (\mathsf{eval}_{N+1} )_*\int \mu_{N+1}^{\nu}({\rm d}\lambda)\mathsf{n}_{\lambda}\circ \Lambda_N^{N+1}(\cdot) .
\end{gather*}
and so we finally arrive at
\begin{gather*}
\big(\mathsf{eval}\circ\pi_N^{N+1}\big)_*\mathsf{M}^{(\nu),N+1}= (\mathsf{eval}_{N+1} )_*\mathsf{M}^{(\nu),N+1}\circ \Lambda_N^{N+1}.\tag*{\qed}
\end{gather*}\renewcommand{\qed}{}
\end{proof}

\begin{Remark}The exact same scheme of proof, can be applied to the case of the Hua--Pickrell measures. One uses, instead the backward shift equation for the pseudo-Jacobi polynomials, see \cite[Section~9.9]{HypergeometricOrthogonalPolynomials}.
\end{Remark}

We finally reinterpret the proposition above to obtain the following matrix integral. This is an analogue of Hua's integrals \cite{Hua,Neretin} for the inverse Wishart measures. Let
\begin{gather*}
f_{N}^{\nu}(X)=\det(X)^{-\nu-2N} {\rm e}^{-2\operatorname{Tr}X^{-1}}\mathbf{1}_{\{X \in H_+(N)\}}.
\end{gather*}
For $X\in H(N+1)$ we write
\begin{gather*}
X=\begin{bmatrix}
\tilde{X} & \zeta \\
\zeta^* & s
\end{bmatrix},
\end{gather*}
where $\tilde{X}=\pi_N^{N+1}(X) \in H(N)$, $\zeta \in \mathbb{C}^{N}$, $s \in \mathbb{R}$. Note that, $\{X \in H_+(N+1)\}$ implies $\{\tilde{X} \in H_+(N)\}$.
\begin{Corollary}
For $\nu>-1$ and $\forall\, N \ge 1$
\begin{gather*}
\int_{(\zeta,s)\in \mathbb{C}^N\times \mathbb{R}}^{}f_{N+1}^{\nu}\left(\begin{bmatrix}
\tilde{X} & \zeta \\
\zeta^* & s
\end{bmatrix}\right)\prod_{i=1}^{N}{\rm d} \operatorname{Re} \zeta_i {\rm d}\operatorname{Im} \zeta_i {\rm d}s=\frac{\operatorname{const}_{\nu, N+1}}{\operatorname{const}_{\nu, N}}f_{N}^{\nu}\big(\tilde{X}\big).
\end{gather*}
\end{Corollary}

\section{Approximation of spectral measures}\label{section3}
In order to proceed we will need to get a handle on the abstract spectral measure $\mathfrak{m}^{(\nu)}$. The following approximation procedure of Olshanski--Vershik~\cite{OlshanskiVershik} and Borodin--Olshanski~\cite{BorodinOlshanskiErgodic}, based on the ergodic method of Vershik~\cite{VershikErgodic} allows us to do so.

For $\lambda^{(N)} \in W^N$ we define the numbers
\begin{gather*}
\alpha_{i,N}^+\big(\lambda^{(N)}\big) =\begin{cases}
\dfrac{\max\big\{\lambda^{(N)}_{i},0\big\}}{N},& i=1,\dots, N,\\
0 & i=N+1,N+2,\dots,
\end{cases} \\
\alpha_{i,N}^-\big(\lambda^{(N)}\big) =\begin{cases}
\dfrac{\max\big\{{-}\lambda^{(N)}_{N+1-i},0\big\}}{N}, & i=1,\dots, N,\\
0 & i=N+1,N+2,\dots.
\end{cases}
\end{gather*}

Equivalently, if $k$ and $l$ denote the number of strictly positive terms in $\big\{\alpha_{i,N}^+\big(\lambda^{(N)}\big) \big\}$ and $ \big\{\alpha_{i,N}^-\big(\lambda^{(N)}\big) \big\}$ respectively, we have
\begin{gather*}
\frac{\lambda^{(N)}}{N}=\big(\alpha_{1,N}^+\big(\lambda^{(N)}\big),\dots,\alpha_{k,N}^+\big(\lambda^{(N)}\big),0, \dots, 0, -\alpha_{l,N}^-\big(\lambda^{(N)}\big),\dots,-\alpha_{1,N}^-\big(\lambda^{(N)}\big)\big).
\end{gather*}

 Define the corner map $\pi_N^{\infty}\colon H \to H(N)$ by $\pi_N^{\infty}(X)=(X_{ij})^N_{j,j=1}$. Let $X\in H$ be given. Then, we define the numbers $\big\{\alpha_{i,N}^+(X) \big\}$, $\big\{\alpha_{i,N}^-(X) \big\}$ by
\begin{gather*}
\alpha_{i,N}^+(X)=\alpha_{i,N}^+\big(\mathsf{eval}_N\big(\pi_N^{\infty}(X)\big)\big), \qquad \forall\, i \ge 1,\\
\alpha_{i,N}^-(X)=\alpha_{i,N}^-\big(\mathsf{eval}_N\big(\pi_N^{\infty}(X)\big)\big), \qquad \forall\, i \ge 1.
\end{gather*}
We also set,
\begin{gather*}
c^{(N)}(X) =\frac{\operatorname{Tr}[\pi_N^{\infty}(X)]}{N}=\sum_{i}^{}\alpha_{i,N}^+(X) -\sum_{i}^{}\alpha_{i,N}^-(X),\\
d^{(N)}(X) =\frac{\operatorname{Tr}[\pi_N^{\infty}(X)^2]}{N^2}=\sum_{i}^{}\big(\alpha_{i,N}^+(X)\big)^2+\sum_{i}^{}\big(\alpha_{i,N}^-(X)\big)^2 .
\end{gather*}

Having all these notations in place we now arrive at the following important definition (see Theorem \ref{ErgodicMethodTheorem} below for the motivation behind it):

\begin{Definition}A matrix $X \in H$ is called regular and we will write $X \in H_{{\rm reg}}$ if the following limits exist
\begin{gather*}
\alpha_i^{\pm}(X) =\lim_{N \to \infty}\alpha_{i,N}^{\pm}(X), \qquad \forall\, i \ge 1,\\
\gamma_1(X) =\lim_{N \to \infty}c^{(N)}(X),\\
\delta(X) =\lim_{N \to \infty}d^{(N)}(X).
\end{gather*}
We can easily see that $\sum_{i}^{} (\alpha_{i}^+(X) )^2+ (\alpha_{i}^-(X) )^2 \le \delta(X)$ and we define
\begin{gather*}
\gamma_2(X)=\delta(X)-\sum_{i}^{}\big(\alpha_{i}^+(X)\big)^2-\sum_{i}^{}\big(\alpha_{i}^-(X)\big)^2\ge 0.
\end{gather*}
\end{Definition}

We also define $\mathfrak{r}_N\colon H \to \Omega$ by
\begin{gather*}
\mathfrak{r}_N(X)=\big(\big\{\alpha_{i,N}^{+}(X)\big\}, \big\{\alpha_{i,N}^{-}(X)\big\}, c^{(N)}(X), d^{(N)}(X)\big)
\end{gather*}
and similarly $\mathfrak{r}_{\infty}\colon H \to \Omega$, that is defined correctly on $H_{{\rm reg}}$ and thus almost everywhere on $H$ as we shall see in Theorem~\ref{ErgodicMethodTheorem} below
\begin{gather*}
\mathfrak{r}_\infty(X)=\big(\big\{\alpha_{i}^{+}(X)\big\}, \big\{\alpha_{i}^{-}(X)\big\}, \gamma_1(X), \delta(X)\big).
\end{gather*}

With all these definitions in place we can now state the following result from \cite[Section 5]{BorodinOlshanskiErgodic}.

\begin{Theorem}\label{ErgodicMethodTheorem}Let $\mathfrak{M}$ be any $\mathbb{U}(\infty)$-invariant probability measure on $H$. Then, $\mathfrak{M}$ is supported on $H_{{\rm reg}}$. Moreover, there exists a unique spectral measure $\mathbb{P}^{\mathfrak{M}}$ associated to~$\mathfrak{M}$ $($in the particular case $\mathfrak{M}=\mathsf{M}^{(\nu)}$ we write as we did in the introduction $\mathfrak{m}^{(\nu)}=\mathbb{P}^{\mathsf{M}^{(\nu)}})$ defined by
\begin{gather*}
\mathfrak{M}=\int_{\Omega}^{}M_{\omega} \mathbb{P}^{\mathfrak{M}}({\rm d}\omega),
\end{gather*}
where the equality is understood as follows: for all Borel functions $\mathfrak{F}$ on~$H$
\begin{gather*}
\int_{H}^{}\mathfrak{F}(X)\mathfrak{M}({\rm d}X)=\int_{\Omega}^{}\int_{H}^{}\mathfrak{F}(X)M_{\omega}({\rm d}X)\mathbb{P}^{\mathfrak{M}}({\rm d}\omega).
\end{gather*}
Furthermore, the spectral measure $\mathbb{P}^{\mathfrak{M}}$ is given explicitly by
\begin{gather*}
\mathbb{P}^{\mathfrak{M}}= (\mathfrak{r}_{\infty} )_* (\mathfrak{M}|_{H_{{\rm reg}}} ).
\end{gather*}
Finally, we have the following weak convergence of probability measures
\begin{gather*}
 (\mathfrak{r}_N )_* \mathfrak{M} \implies (\mathfrak{r}_{\infty} )_* (\mathfrak{M}|_{H_{{\rm reg}}} )=\mathbb{P}^{\mathfrak{M}}.
\end{gather*}
\end{Theorem}

We will need some more definitions and notations used in later sections. For $X$ (deterministic or random) in $H$ and $H_{{\rm reg}}$ respectively define the point configurations
\begin{gather*}
\mathsf{C}_N(X) =\big\{\alpha_{i,N}^{+}(X)\big\} \sqcup \big\{{-}\alpha_{i,N}^{-}(X)\big\}, \\
\mathsf{C}(X) =\big\{\alpha_{i}^{+}(X)\big\} \sqcup \big\{{-}\alpha_{i}^{-}(X)\big\},
\end{gather*}
omitting possible zeroes. We will write $\mathsf{C}^{(\nu)}_N(X), \mathsf{C}^{(\nu)}(X)$ if $\mathsf{Law}(X)=\mathsf{M}^{(\nu)}$.

For $\mathfrak{M}$ on $H$ that is $\mathbb{U}(\infty)$-invariant we will let $\mathcal{P}_N$ and $\mathcal{P}$ (we drop dependence on~$\mathfrak{M}$) denote the corresponding random point configurations, namely $\mathsf{C}_N(X)$ and $\mathsf{C}(X)$ if $\mathsf{Law}(X)=\mathfrak{M}$.

We need a final definition for a general measure $\mathbb{P}$ on $\mathsf{Conf}(\mathcal{X})$. We define its correlation functions $\rho_n$ for each $n\ge 1$ (if they exist) with respect to $\mu$, by replacing the r.h.s.\ of~(\ref{determinantaldefinition}) by
\begin{gather*}
\int_{\mathcal{X}^n}^{}\Phi(x_1,\dots,x_n) \rho_n(x_1,\dots,x_n){\rm d}\mu(x_1) \cdots {\rm d}\mu(x_n).
\end{gather*}
In particular, for a determinantal measure $\mathbb{P}_{\mathsf{K}}$ we have $\forall\, n \ge 1$
\begin{gather*}
\rho_n(x_1,\dots,x_n)=\det (\mathsf{K}(x_i,x_j) )_{i,j=1}^n.
\end{gather*}

\section[Limit of the correlation kernel and the $\alpha$ parameters]{Limit of the correlation kernel and the $\boldsymbol{\alpha}$ parameters}\label{section4}
\subsection[The $\alpha^-$ parameters]{The $\boldsymbol{\alpha^-}$ parameters}\label{section4.1}
The following proposition is obvious from the definitions in Section~\ref{section3} and Theorem~\ref{ErgodicMethodTheorem}, since the measure~$\mathsf{M}^{(\nu)}$ is supported on~$H_+$:
\begin{Proposition}\label{PropositionAlphaMinus}
Let $\nu>-1$. Then,
\begin{gather*}
\alpha_i^-(X)=0,\qquad \forall\, i\ge 1 \qquad \text{for} \quad \mathsf{M}^{(\nu)}-\operatorname{a.e.} X \in H_{{\rm reg}}.
\end{gather*}
\end{Proposition}

Thus, we can restrict our attention on the parameters $\alpha^+$ which form a random point configuration on $\mathcal{X}=(0,\infty)$ (throughout, the reference measure $\mu$ on $(0,\infty)$ will be Lebesgue measure), which we go on to study next.

\subsection[Explicit expression for the correlation kernel and the $\alpha^+$ parameters]{Explicit expression for the correlation kernel and the $\boldsymbol{\alpha^+}$ parameters}\label{section4.2}

 It is a standard result from random matrix theory that the measure $\mu^{\nu}_N$ gives rise to a determinantal point process on $(0,\infty)$ with $N$ particles. It is the orthogonal polynomial ensemble associated to the Bessel weight~$w_N^{\nu}$. We will denote its correlation kernel by~$\mathsf{K}_N^{\nu}$. This is given explicitly in terms of Bessel polynomials but for the purposes of the asymptotic analysis we will instead exploit the connection to the Laguerre ensemble.

First, we need to recall a simple fact about transformations of determinantal measures on $\mathcal{X}=(I_1,I_2) \subset \mathbb{R}$. Let $f\colon \mathcal{X} \to \mathcal{X}$ be a $C^1$ bijection and let $g=f^{-1}$ be its inverse. Then, $f$~induces a homeomorphism of $\mathsf{Conf}(\mathcal{X})$: for a configuration $\mathfrak{X}$ the particles of $f(\mathfrak{X})$ are of the form~$f(x)$, $x \in \mathfrak{X}$. Let $\mathbb{P}_{\mathsf{K}}$ be a determinantal measure on $\mathcal{X}$. Then, its pushforward under~$f$ is also a~determinantal measure
\begin{gather*}
f_*\mathbb{P}_{\mathsf{K}}=\mathbb{P}_{\hat{\mathsf{K}}^f}
\end{gather*}
with correlation kernel
\begin{gather*}
\hat{\mathsf{K}}^f(x,y)=\sqrt{g'(x)g'(y)}\mathsf{K}(g(x),g(y)).
\end{gather*}

The Laguerre ensemble $\mathfrak{L}_N^{\nu}$, defined for $\nu>-1$, is the probability measure on $W^N_+$
\begin{gather*}
\mathfrak{L}_N^{\nu}({\rm d}x)=\tilde{c}_{\nu,N} \prod_{i=1}^{N}x_i^{\nu}{\rm e}^{-x_i}\Delta_N(x)^2 \textbf{1}_{\{x \in W^N_+\}} {\rm d}x_1\cdots {\rm d}x_N.
\end{gather*}

We will denote for $n\ge 1$, by $L_n^{\nu}$ the monic Laguerre polynomials, orthogonal with respect to the measure (note that this, unlike $w^{\nu}_N$, does not depend on~$N$)
\begin{gather*}
\lambda_\nu(x){\rm d}x=x^{\nu} {\rm e}^{-x}\textbf{1}_{\{x>0\}}{\rm d}x.
\end{gather*}
Their squared norm is given by:
\begin{gather*}
\|L_n^{\nu}\|_2^2=n! \Gamma(n+\nu+1).
\end{gather*}

It is a standard result that the Laguerre ensemble $\mathfrak{L}_N^{\nu}({\rm d}x)$ gives rise to a determinantal point process on $(0,\infty)$ with $N$ particles and its correlation kernel $\mathsf{L}_{N}^{\nu}$, with respect to Lebesgue measure, is given by
\begin{gather*}
\mathsf{L}_{N}^{\nu}(x,y)=\sum_{i=0}^{N-1}\frac{L_i^{\nu}(x)L_i^{\nu}(y)}{\|L_n^{\nu}\|_2^2}\sqrt{\lambda_\nu(x)\lambda_\nu(y)}.
\end{gather*}
Using the Christoffel--Darboux formula we can further write
\begin{gather*}
\mathsf{L}_{N}^{\nu}(x,y)=\frac{1}{\|L_{N-1}^{\nu}\|_2^2}\frac{L_{N}^{\nu}(x)L_{N-1}^{\nu}(y)-L_{N-1}^{\nu}(x)L_{N}^{\nu}(y)}{x-y}\sqrt{\lambda_\nu(x)\lambda_\nu(y)}.
\end{gather*}
Now, note that under the transformation $f(x)=\frac{2}{x}$ we have
\begin{gather*}
(f_*\mathfrak{L}_N^{\nu})({\rm d}x)=\mu_N^{\nu}({\rm d}x), \qquad \forall\, N \ge 1.
\end{gather*}
Thus, the correlation kernel $\mathsf{K}_N^{\nu}$ of the determinantal point process associated to $\mu_N^{\nu}$ is given by
\begin{gather*}
\mathsf{K}_N^{\nu}(x,y)=\frac{2}{xy}\mathsf{L}_N^{\nu}\left(\frac{2}{x},\frac{2}{y}\right).
\end{gather*}
Furthermore, the scaled point process $\mathsf{C}_N^{(\nu)}(X)=\{\alpha_{i,N}^+(X) \}_{i=1}^N$ where $\mathsf{Law}(X)=\mathsf{M}^{(\nu)}$ is deter\-minantal on $(0,\infty)$ with correlation kernel $K_N^{\nu}$, with respect to Lebesgue measure, given by
\begin{gather*}
K_N^{\nu}(x,y)=N\mathsf{K}_N^{\nu}(Nx,Ny).
\end{gather*}
And also in terms of the Laguerre kernel
\begin{gather}\label{KernelRelation}
K_N^{\nu}(x,y)=\frac{2}{Nxy}\mathsf{L}_N^{\nu}\left(\frac{2}{Nx},\frac{2}{Ny}\right).
\end{gather}

\begin{Proposition}\label{PropositionLimitCorrelationAlphaPlus}
Let $\nu>-1$. Then, we have uniformly on compacts in $(0,\infty)$
\begin{gather*}
K_N^{\nu}(x,y)\underset{N\to \infty}{\longrightarrow}K_\infty^{\nu}(x,y)=\frac{8}{xy}\mathbb{J}_{\nu}\left(\frac{8}{x},\frac{8}{y}\right).
\end{gather*}
Furthermore, the determinantal point process $\mathsf{C}^{(\nu)}_{N}(X)=\{\alpha_{i,N}^+(X) \}$ with $\mathsf{Law}(X)=\mathsf{M}^{(\nu)}$ on $(0,\infty)$ convergences weakly as $N \to \infty$ to the determinantal point process $\mathsf{C}^{(\nu)}(X)=\{ \alpha_{i}^+(X) \}$ such that
\begin{gather*}
\mathsf{Law}\big(\mathsf{C}^{(\nu)}(X)\big)=\mathbb{P}_{K_{\infty}^{\nu}}
\end{gather*}
and so do all the correlation functions.
Finally, under the transformation $x \mapsto \frac{8}{x}$ we get the Bessel point process
\begin{gather*}
\mathsf{Law}\left(\frac{8}{\mathsf{C}^{(\nu)}(X)}\right)=\mathbb{P}_{\mathbb{J}_{\nu}}.
\end{gather*}
\end{Proposition}

\begin{proof}
We first write using relation (\ref{KernelRelation})
\begin{gather*}
K_N^{\nu}(x,y)=\frac{8}{xy}\frac{1}{4N}\mathsf{L}_{N}^{\nu}\left(\frac{1}{4N}\frac{8}{x},\frac{1}{4N}\frac{8}{y}\right).
\end{gather*}
Then, the first statements of the proposition are immediate consequences of the following well-known facts: the uniform convergence on compacts in $(0,\infty)$, with $z_1=\frac{8}{x}, z_2=\frac{8}{y}$
\begin{gather*}
\lim_{N \to \infty}\frac{1}{4N}\mathsf{L}_N^{\nu}\left(\frac{1}{4N}z_1,\frac{1}{4N}z_2\right)=\mathbb{J}_{\nu}(z_1,z_2)
\end{gather*}
and convergence of the Laguerre ensemble determinantal point process at the hard edge scaling limit to the Bessel process, see \cite{ForresterBessel,Forrester, TracyWidom}. These results follow from uniform on compacts asymptotics for Laguerre polynomials, see~\cite{Szego}. The final statement is immediate from the transformation rule for determinantal point processes.
\end{proof}

\begin{Remark}Using, the representation of $\mathsf{K}_N^{\nu}$ and thus $K_N^\nu$ in terms of the Bessel polynomials it is possible to give an alternative proof of Proposition~\ref{PropositionLimitCorrelationAlphaPlus}. The analysis boils down to asymptotics for hypergeometric functions.
\end{Remark}

Finally, the following completes the description of the $\alpha$ parameters.

\begin{Proposition}\label{PropositionAlphaPlus} Let $\nu>-1$. Then,
\begin{gather*}
\alpha_i^+(X)\neq 0, \qquad \forall\, i\ge 1, \qquad \text{for}\quad \mathsf{M}^{(\nu)}-\operatorname{a.e.} X \in H_{{\rm reg}}.
\end{gather*}
\end{Proposition}

\begin{proof}Observe that, since the $\alpha_i^+(X)$ are strictly decreasing (by Proposition~\ref{PropositionLimitCorrelationAlphaPlus} they form a~determinantal point process) it suffices to prove that
\begin{gather*}
\mathsf{C}^{(\nu)}(X)\  \text{has infinitely many points for} \ \mathsf{M}^{(\nu)}-\operatorname{a.e.} X \in H_{{\rm reg}}.
\end{gather*}
Under the map $x \mapsto \frac{8}{x}$, by Proposition~\ref{PropositionLimitCorrelationAlphaPlus}, it then suffices to prove that the Bessel point process~$\mathbb{P}_{\mathbb{J}_{\nu}}$ has infinitely many particles almost surely which is a well-known result (in fact stronger quantitative results are known on the number of particles in growing intervals, see \cite[Theorem~2]{SoshnikovGaussianFluctuations}). This is a simple consequence of Theorem~4 of~\cite{Soshnikov} and the fact that $\mathbb{J}_{\nu}$ is a projection kernel of infinite rank (more precisely~$\mathbb{J}_{\nu}$ induces on $L^2 ((0,\infty),{\rm d}x )$ the operator of orthogonal projection onto the subspace of functions whose Hankel transform is supported on~$[0,1]$, see~\cite{TracyWidom}).
\end{proof}

\section{An estimate on the correlation kernel}\label{section5}
In this section we obtain certain, uniform in $N$, estimates on $K_N^{\nu}$, that will be useful for the description of the parameters $\gamma_1$ and $\gamma_2$.

\begin{Proposition}\label{EstimateGamma1}
Let $\nu>-1$. For any $\epsilon>0$, there exists $\delta>0$ such that for all $N \in \mathbb{N}$
\begin{gather*}
\int_{0}^{\delta}xK^{\nu}_N(x,x){\rm d}x<\epsilon.
\end{gather*}
\end{Proposition}

Using relation (\ref{KernelRelation}), then in terms of the Laguerre kernel $\mathsf{L}_N^{\nu}$ it will suffice to prove:

\begin{Proposition}\label{EstimateLaguerre}
Let $\nu>-1$. For any $\epsilon>0$ there exists $R=R(\epsilon)$ large enough such that for all $N \in \mathbb{N}$
\begin{gather*}
\int_{\frac{R}{N}}^{\infty}\frac{1}{N}\frac{\mathsf{L}_N^{\nu}(y,y)}{y}{\rm d}y<\epsilon.
\end{gather*}
\end{Proposition}

{\bf Bounds for Laguerre polynomials.} In order to prove Proposition~\ref{EstimateLaguerre} we first need to recall some bounds for Laguerre polynomials from the classical book of Szeg\"{o}~\cite{Szego}.

\begin{Lemma}\label{wavefunctionestimate}We have the following, uniform in $n$, estimate for the Laguerre wavefunction $\mathcal{W}_n^{\nu}$ $($defined by the first equality below$)$, with $\mathsf{w}\ge 1$ being an arbitrary fixed constant
\begin{gather*}
\mathcal{W}_{n}^{\nu}(x)=\frac{L_n^{\nu}(x)^2x^{\nu}{\rm e}^{-x}}{\|L_n^{\nu}\|_2^2}\le \begin{cases}
\mathsf{c}(\mathsf{w})\times x^{-\frac{1}{2}}n^{-\frac{1}{2}}, & n^{-1}\le x \le \mathsf{w},\\
C\times x^\nu n^{\nu}, & 0 \le x \le n^{-1}.
\end{cases}
\end{gather*}
Here, $\mathsf{c}(\mathsf{w})$ only depends on $\mathsf{w}$ and in particular is independent of $n$, while $C$ is a generic constant independent of all quantities involved in the statement above.
\end{Lemma}

\begin{proof}We make use of results from Szeg\"{o}~\cite{Szego}. We note that in~\cite{Szego} statements involving Laguerre polynomials are for the ones normalized to have leading coefficient~$\frac{(-1)^n}{n!}$. Since we are interested in the \textit{monic} Laguerre polynomials $L_n^{\nu}$ we simply need to multiply the formulae therein by the inverse of this coefficient. Then, Theorem~7.6.4 in~\cite{Szego} reads as follows in our setting, for some fixed constant $\mathsf{w}\ge 1$ and uniformly in~$n$
\begin{gather*}
L_n^{\nu}(x)\le \begin{cases}
\tilde{\mathsf{c}}(\mathsf{w})\times n! x^{-\frac{\nu}{2}-\frac{1}{4}}n^{\frac{\nu}{2}-\frac{1}{4}}, & n^{-1}\le x \le \mathsf{w},\\
\tilde{C} \times n! n^{\nu}, & 0 \le x \le n^{-1}.
\end{cases}
\end{gather*}
Here, $\tilde{\mathsf{c}}(\mathsf{w})$ only depends on $\mathsf{w}$ and is independent of~$n$, while~$\tilde{C}$ is a generic constant independent of all quantities involved in the statement above.

Now, recall that $\|L_n^{\nu}\|_2^2=n! \Gamma(n+\nu+1)$. Then, using the classical fact for ratios of Gamma functions
\begin{gather*}
\frac{n!}{\Gamma(n+\nu+1)}=\frac{\Gamma(n+1)}{\Gamma(n+1+\nu)}= n^{-\nu}+\mathcal{O}\left(\frac{1}{n^{\nu+1}}\right)
\end{gather*}
and the trivial bound ${\rm e}^{-x}\le 1$ we obtain the statement of the~lemma.
\end{proof}

With these preliminaries in place we are ready to prove Proposition~\ref{EstimateLaguerre}.

\begin{proof}[Proof of Proposition~\ref{EstimateLaguerre}]
Let $\epsilon>0$ be arbitrary (and for convenience assume $\epsilon<1$). Below, we will pick (in this order) constants $\mathsf{w}, l, R$ depending on $\epsilon$. This choice will be made in display~(\ref{ChoiceOfConstants}) based on the bound~(\ref{MainEstimate}).

We will use the notation $\lesssim$ to mean $\le$ up to a constant; with the implicit constant being independent of $\epsilon$ (and obviously of the quantities $\mathsf{w}$, $l$, $R$ depending on it) and uniform in~$N$ (in particular we keep track of the constant $\mathsf{c}(\mathsf{w})$ from Lemma~\ref{wavefunctionestimate} but not of the generic constant~$C$). It is clear that it suffices to prove that we can find $R$ large enough such that uniformly in $N \in \mathbb{N}$:
\begin{gather*}
\int_{\frac{R}{N}}^{\infty}\frac{1}{N}\frac{\mathsf{L}_N^{\nu}(y,y)}{y}{\rm d}y \lesssim \epsilon.
\end{gather*}
First, recall that
\begin{gather*}
\mathsf{L}_N^{\nu}(y,y)=\sum_{k=0}^{N-1}\mathcal{W}_k^\nu(y).
\end{gather*}
We split the integral, with $\mathsf{w}$ to be picked according to $\epsilon$ (large)
\begin{gather*}
\int_{\frac{R}{N}}^{\infty}\frac{1}{N}\frac{\mathsf{L}_N^{\nu}(y,y)}{y}{\rm d} y= \int_{\frac{R}{N}}^{\mathsf{w}}\frac{1}{N}\frac{\mathsf{L}_N^{\nu}(y,y)}{y}{\rm d}y +\int_{\mathsf{w}}^{\infty}\frac{1}{N}\frac{\mathsf{L}_N^{\nu}(y,y)}{y}{\rm d}y.
\end{gather*}
We then, using the fact that $\int_{0}^{\infty}\mathcal{W}_n^{\nu}(y){\rm d}y=1$ for all $n$, easily estimate
\begin{gather*}
\frac{1}{N}\int_{\mathsf{w}}^{\infty}\frac{\mathsf{L}_N^{\nu}(y,y)}{y}{\rm d}y\le\frac{1}{\mathsf{w}N} \int_{\mathsf{w}}^{\infty}\mathsf{L}_N^{\nu}(y,y){\rm d}y \le \frac{1}{\mathsf{w}N}N=\frac{1}{\mathsf{w}}.
\end{gather*}
We now focus on the integral
\begin{gather*}
\int_{\frac{R}{N}}^{\mathsf{w}}\frac{1}{N}\frac{\mathsf{L}_N^{\nu}(y,y)}{y}{\rm d}y =\frac{1}{N}\sum_{1 \le n < N}^{} \int_{\frac{R}{N}}^{\mathsf{w}}\frac{\mathcal{W}_n^{\nu}(y)}{y}{\rm d}y
\end{gather*}
and split the range of summation, for $l$ to be picked depending on $\epsilon$ (small), as follows
\begin{gather*}
\frac{1}{N}\sum_{1 \le n < lN}^{} \int_{\frac{R}{N}}^{\mathsf{w}}\frac{\mathcal{W}_n^{\nu}(y)}{y}{\rm d}y+\frac{1}{N}\sum_{l N \le n < N}^{} \int_{\frac{R}{N}}^{\mathsf{w}}\frac{\mathcal{W}_n^{\nu}(y)}{y}{\rm d}y.
\end{gather*}
We will moreover, in the range $1 \le n < l N$, split the integral further as
\begin{gather*}
\int_{\frac{R}{N}}^{\mathsf{w}}\frac{\mathcal{W}_n^{\nu}(y)}{y}{\rm d}y =\int_{\frac{R}{N}}^{\frac{1}{n}}\frac{\mathcal{W}_n^{\nu}(y)}{y}{\rm d}y+\int_{\frac{1}{n}}^{\mathsf{w}}\frac{\mathcal{W}_n^{\nu}(y)}{y}{\rm d}y.
\end{gather*}
Note that, since the integrand is positive, in case $\frac{R}{N}>\frac{1}{n}$, we simply estimate
\begin{gather*}
\int_{\frac{R}{N}}^{\mathsf{w}}\frac{\mathcal{W}_n^{\nu}(y)}{y}{\rm d}y \le \int_{\frac{1}{n}}^{\mathsf{w}}\frac{\mathcal{W}_n^{\nu}(y)}{y}{\rm d}y.
\end{gather*}
Thus, we have the bound
\begin{gather*}
\int_{\frac{R}{N}}^{\infty}\frac{1}{N}\frac{\mathsf{L}_N^{\nu}(y,y)}{y}{\rm d}y \le \frac{1}{\mathsf{w}}+J_1+J_2+J_3,
\end{gather*}
where we define
\begin{gather*}
J_1 = \frac{1}{N}\sum_{1 \le n < lN}^{} \textbf{1}_{\left(\frac{1}{n}\ge \frac{R}{N}\right)} \int_{\frac{R}{N}}^{\frac{1}{n}}\frac{\mathcal{W}_n^{\nu}(y)}{y}{\rm d}y,\\
J_2 = \frac{1}{N}\sum_{1 \le n < lN}^{} \int_{\frac{1}{n}}^{\mathsf{w}}\frac{\mathcal{W}_n^{\nu}(y)}{y}{\rm d}y,\qquad
J_3 = \frac{1}{N}\sum_{l N \le n < N}^{} \int_{\frac{R}{N}}^{\mathsf{w}}\frac{\mathcal{W}_n^{\nu}(y)}{y}{\rm d}y.
\end{gather*}
We now go on to estimate each of these terms individually. Using the bound for $\mathcal{W}_n^{\nu}$ from Lemma~\ref{wavefunctionestimate} in the range $0 \le x \le n^{-1}$, we estimate $J_1$:
\begin{gather*}
J_1 \lesssim \frac{1}{N} \sum_{1 \le n < lN}^{}\textbf{1}_{\left(\frac{1}{n}\ge \frac{R}{N}\right)} n^{\nu} \int_{\frac{R}{N}}^{\frac{1}{n}}y^{\nu-1}{\rm d}y.
\end{gather*}
We split into three cases. For $\nu>0$ we have
\begin{gather*}
J_1 \lesssim \frac{1}{N} \sum_{1 \le n < lN}^{}\textbf{1}_{\left(\frac{1}{n}\ge \frac{R}{N}\right)}n^{\nu} \left[\left(\frac{1}{n}\right)^\nu-\left(\frac{R}{N}\right)^\nu\right]\lesssim l.
\end{gather*}
While, for $\nu=0$
\begin{align*}
J_1 & \lesssim \frac{1}{N} \sum_{1 \le n <lN}^{}\textbf{1}_{\left(\frac{1}{n}\ge \frac{R}{N}\right)} \log \left(\frac{N}{nR}\right)=\frac{1}{N}\sum_{1 \le n <lN}^{}\left[-\textbf{1}_{\left(\frac{1}{n}\ge \frac{R}{N}\right)} \log \left(\frac{n}{N}\right)-\textbf{1}_{\left(\frac{1}{n}\ge \frac{R}{N}\right)} \log \left(R\right)\right]\\
 &\le \frac{1}{N}\sum_{1 \le n <lN}^{}- \log \left(\frac{n}{N}\right) \lesssim \int_{0}^{l}-\log(x){\rm d}x=-l\log(l)-l.
\end{align*}
Finally, for $-1<\nu<0$
\begin{gather*}
J_1 \lesssim \frac{1}{N} \sum_{1 \le n <lN}^{}\textbf{1}_{\left(\frac{1}{n}\ge \frac{R}{N}\right)} \left(\frac{Rn}{N}\right)^\nu \lesssim \frac{1}{N^{\nu+1}}\left(lN\right)^{\nu+1}R^\nu=l^{\nu+1}R^\nu.
\end{gather*}
Observe that, the last bound decreases as $R$ increases since $\nu$ is negative.

Now, using the bound for $\mathcal{W}_n^{\nu}$ from Lemma~\ref{wavefunctionestimate} in the range $n^{-1} \le x \le \mathsf{w}$, we estimate $J_2$:
\begin{gather*}
J_2 \lesssim \frac{\mathsf{c}(\mathsf{w})}{N} \sum_{1 \le n < lN}^{} \int_{\frac{1}{n}}^{\mathsf{w}} \frac{1}{y^{\frac{3}{2}}}n^{-\frac{1}{2}}{\rm d}y \lesssim \frac{\mathsf{c}(\mathsf{w})}{N} \sum_{1 \le n < lN}^{}n^{-\frac{1}{2}} \left[n^{\frac{1}{2}}-\frac{1}{\mathsf{w}^{\frac{1}{2}}}\right] \lesssim \mathsf{c}(\mathsf{w})l.
\end{gather*}

Finally, we turn to $J_3$. We assume that $R$ is large enough so that $R \ge \frac{1}{l}$. The fact that this is possible (not trivial apriori since both quantities depend on $\epsilon$) will be clear by the choices made in (\ref{ChoiceOfConstants}) below. Thus, the bound for $\mathcal{W}_n^{\nu}$ from Lemma \ref{wavefunctionestimate} valid in the range $n^{-1}\le x \le \mathsf{w}$, is valid throughout the range of integration $\frac{R}{N} \le x \le \mathsf{w}$. Hence, we estimate
\begin{align*}
J_3 &\lesssim \frac{\mathsf{c}(\mathsf{w})}{N} \sum_{lN \le n < N}^{}\int_{\frac{R}{N}}^{\mathsf{w}}n^{-\frac{1}{2}}\frac{1}{y^{\frac{3}{2}}}{\rm d}y\lesssim \frac{\mathsf{c}(\mathsf{w})}{N^{\frac{3}{2}}l^{\frac{1}{2}}} \sum_{lN \le n < N}^{}\int_{\frac{R}{N}}^{\mathsf{w}}\frac{1}{y^{\frac{3}{2}}}{\rm d}y\le \frac{\mathsf{c}(\mathsf{w})}{N^{\frac{3}{2}}l^{\frac{1}{2}}} \sum_{lN \le n < N}^{}\int_{\frac{R}{N}}^{\infty}\frac{1}{y^{\frac{3}{2}}}{\rm d}y \nonumber\\
&\lesssim \frac{\mathsf{c}(\mathsf{w})}{N^{\frac{3}{2}}l^{\frac{1}{2}}} N \times \frac{N^{\frac{1}{2}}}{R^{\frac{1}{2}}}=\frac{\mathsf{c}(\mathsf{w})}{R^{\frac{1}{2}}l^{\frac{1}{2}}}.
\end{align*}
Putting everything together we obtain
\begin{gather}\label{MainEstimate}
\int_{\frac{R}{N}}^{\infty}\frac{1}{N}\frac{\mathsf{L}_N^{\nu}(y,y)}{y}{\rm d}y \lesssim \begin{cases}
\dfrac{1}{\mathsf{w}}+l+\mathsf{c}(\mathsf{w})l+\dfrac{\mathsf{c}(\mathsf{w})}{R^{\frac{1}{2}}l^{\frac{1}{2}}}, &\text{for } \nu>0,\vspace{1mm}\\
\dfrac{1}{\mathsf{w}}-l\log(l)-l+\mathsf{c}(\mathsf{w})l+\dfrac{\mathsf{c}(\mathsf{w})}{R^{\frac{1}{2}}l^{\frac{1}{2}}}, &\text{for } \nu=0,\vspace{1mm}\\
\dfrac{1}{\mathsf{w}}+l^{\nu+1}R^{\nu}+\mathsf{c}(\mathsf{w})l+\dfrac{\mathsf{c}(\mathsf{w})}{R^{\frac{1}{2}}l^{\frac{1}{2}}}, &\text{for } -1<\nu<0.
\end{cases}
\end{gather}
We now pick the quantities $(\mathsf{w},l,R)$ depending on $\epsilon$ so that each summand in bound~(\ref{MainEstimate}) is of order at most $\epsilon$ (for $\epsilon$ small). We first choose $\mathsf{w}=\mathsf{w}(\epsilon)=\frac{1}{\epsilon}$. It is also convenient to define the following constant
 \begin{gather*}
\mathsf{c}_{\epsilon}=\max \left\{\mathsf{c}\left(\frac{1}{\epsilon}\right),1 \right\}.
 \end{gather*}
Observe that, with this choice of $\mathsf{w}(\epsilon)$ the bound in~(\ref{MainEstimate}) is clearly still valid if we replace $\mathsf{c}(\mathsf{w}(\epsilon))$ by $\mathsf{c}_{\epsilon}$.

Hence, it is not hard to see that we can choose $(\mathsf{w},l,R)$ as follows (of course other choices are possible)
\begin{gather}\label{ChoiceOfConstants}
(\mathsf{w}(\epsilon),l(\epsilon),R(\epsilon))=\begin{cases}
\displaystyle \left(\frac{1}{\epsilon},\frac{\epsilon}{\mathsf{c}_{\epsilon}}, \frac{\mathsf{c}_{\epsilon}^3}{\epsilon^3}\right), &\text{for } \nu>0,\vspace{1mm}\\
\displaystyle\left(\frac{1}{\epsilon},\frac{\epsilon^{\kappa}}{\mathsf{c}_{\epsilon}}, \frac{\mathsf{c}_{\epsilon}^3}{\epsilon^{2+\kappa}}\right), \ \kappa>1, &\text{for } \nu=0,\vspace{1mm}\\
\displaystyle\left(\frac{1}{\epsilon},\frac{\epsilon}{\mathsf{c}_{\epsilon}}, \frac{\mathsf{c}_{\epsilon}^3}{\epsilon^3}\right), &\text{for } -1<\nu<0.
\end{cases}
\end{gather}

Note that, in all cases above the requirement $R\ge \frac{1}{l}$, that was used to bound $J_3$, is satisfied. Thus, if we take $R=R(\epsilon)$ as above, we get uniformly in $N \in \mathbb{N}$
\begin{gather*}
\int_{\frac{R(\epsilon)}{N}}^{\infty}\frac{1}{N}\frac{\mathsf{L}_N^{\nu}(y,y)}{y}{\rm d}y \lesssim \epsilon.
\end{gather*}
The proof is complete.
\end{proof}

The proof of Proposition \ref{EstimateLaguerre}, in particular the inequality in display (\ref{MainEstimate}), also gives the following lemma, written in terms of $K_N^{\nu}$:

\begin{Lemma}\label{LemmaUniformBoundedness}
Let $\nu>-1$. Let $T>0$ be fixed. Then
\begin{gather*}
\int_{0}^{T}xK^{\nu}_N(x,x){\rm d}x \quad \text{is uniformly bounded in $N$}.
\end{gather*}
\end{Lemma}

\begin{Remark}[estimates for Hua--Pickrell measures]\label{EstimatesHuaPickrell}
For the Hua--Pickrell measures an estimate for the corresponding correlation kernel $K_{\rm HP}^{\mathsf{s},N}$ (the analogue of $K_{N}^{\nu}$) also holds (the integral is over the range $(-\epsilon,\epsilon)$ since the $\alpha_{i,N}^-$ in this case are not trivial)
\begin{gather}\label{HuaPickrellGamma2}
\int_{-\epsilon}^{\epsilon}x^2K_{\rm HP}^{\mathsf{s},N}(x,x){\rm d}x.
\end{gather}
This as we see in the next section gives that $\gamma_2 \equiv 0$.

On the other hand, in the analysis of the parameter $\gamma_1$, one is led to consider the following term
\begin{gather}\label{HuaPickrellGamma1}
\int_{-\epsilon}^{\epsilon}xK_{\rm HP}^{\mathsf{s},N}(x,x){\rm d}x.
\end{gather}
For $\mathsf{s} \in \mathbb{R}$, due to the symmetry of the kernel $K_{\rm HP}^{\mathsf{s},N}(-x,-x)=K_{\rm HP}^{\mathsf{s},N}(x,x)$, it is immediate that this term vanishes identically. Then, one is left with terms that can be easily estimated by~(\ref{HuaPickrellGamma2}). In general, the problem of estimating~(\ref{HuaPickrellGamma1}) (when it's not trivial) appears to be open.
\end{Remark}

\section[The parameter $\gamma_2$]{The parameter $\boldsymbol{\gamma_2}$}\label{section6}

We first recall the following result, see \cite[Proposition~7.2]{BorodinOlshanskiErgodic}:

\begin{Proposition}\label{BorodinOlshanskiGamm2}
Let $\mathfrak{M}$ be a $\mathbb{U}(\infty)$-invariant probability measure on $H$. Let $\mathcal{P}_N$ and $\mathcal{P}$ be the corresponding point processes on $\mathbb{R}^*=\mathbb{R}\backslash\{0\}$ defined in Section~{\rm \ref{section3}}. Let $\rho_1^{(N)}$ and $\rho_1$ be the first correlation functions with respect to Lebesgue measure $($assuming they exist$)$ of these point processes. Assume that, for any $\Phi\in C_c (\mathbb{R}^*)$ we have
\begin{gather}\label{ConvergenceCorrelationGamma2}
\int \Phi(x)\rho_1^{(N)}(x){\rm d}x\to \int \Phi(x)\rho_1(x){\rm d}x.
\end{gather}
Finally, assume that
\begin{gather}\label{UniformEstimateGamma2}
\lim_{\epsilon \to 0} \int_{-\epsilon}^{\epsilon}x^2 \rho_1^{(N)}(x){\rm d}x=0, \qquad \text{uniformly in $N$}.
\end{gather}
Then, we have
\begin{gather*}
\gamma_2(X)=0, \qquad \text{for} \quad \mathfrak{M}-\operatorname{a.e.} X \in H_{{\rm reg}}.
\end{gather*}
\end{Proposition}

The proposition above, along with the results of the previous sections, allows us to determine~$\gamma_2$:

\begin{Proposition}\label{PropositionGamma2}
Let $\nu>-1$. Then
\begin{gather*}
\gamma_2(X)=0, \qquad \text{for} \quad \mathsf{M}^{(\nu)}-\operatorname{a.e.} X \in H_{{\rm reg}}.
\end{gather*}
\end{Proposition}

\begin{proof} We apply Proposition \ref{BorodinOlshanskiGamm2} with $\mathfrak{M}=\mathsf{M}^{(\nu)}$. The convergence of correlation functions~(\ref{ConvergenceCorrelationGamma2}) comes from Proposition \ref{PropositionAlphaPlus}. Moreover,
statement (\ref{UniformEstimateGamma2}) above becomes (recall that the first correlation function $\rho_1^{(N)}(x)=K_N^{\nu}(x,x)$ for $x\ge 0$ and vanishes identically for $x<0$):
\begin{gather*}
\lim_{\epsilon \to 0} \int_{0}^{\epsilon}x^2 K_N^{\nu}(x,x){\rm d}x \to 0, \qquad \text{uniformly in $N$},
\end{gather*}
which is an immediate consequence of Proposition~\ref{EstimateGamma1}.
\end{proof}

\section[The parameter $\gamma_1$]{The parameter $\boldsymbol{\gamma_1}$}\label{section7}

\begin{Proposition}\label{PropositionGamma1}
Let $\nu>-1$. Then
\begin{gather*}
\gamma_1(X)=\sum_{i=1}^{\infty}\alpha_i^+(X)<\infty, \qquad \text{for} \quad \mathsf{M}^{(\nu)}-\operatorname{a.e.} X \in H_{{\rm reg}}.
\end{gather*}
\end{Proposition}

Proposition \ref{PropositionGamma1} will be an easy consequence of the next result, Propositions \ref{FinitenessGamma1} and~\ref{EstimateGamma1}.

\begin{Proposition}\label{AbstractGamma1}
Let $\mathfrak{M}$ be a $\mathbb{U}(\infty)$-invariant measure on $H$ such that
\begin{gather*}
\big(\pi_N^{\infty}\big)_*\mathfrak{M} \quad \text{is supported on} \quad H_+(N),\quad \forall\, N \ge 1.
\end{gather*}
In particular,
\begin{gather*}
\alpha_{i,N}^-(X)\equiv 0, \quad \forall\, i\ge 1, \quad N\ge 1, \qquad \alpha_i^-(X)\equiv 0, \quad \forall\, i \ge 1 \quad \text{for} \quad \mathfrak{M}-\operatorname{a.e.} X \in H_{{\rm reg}}.
\end{gather*}
Moreover, assume that
\begin{gather*}
\gamma_1(X)<\infty, \qquad \text{for} \quad \mathfrak{M}-\operatorname{a.e.} X \in H_{{\rm reg}}.
\end{gather*}
Let $\mathcal{P}_N, \mathcal{P}$ be the corresponding point processes $($of $\alpha^+$'s$)$ on $(0,\infty)$ and let $\rho_1^{(N)}$ and $\rho_1$ be their first correlation functions with respect to Lebesgue measure $($assuming they exist$)$. Assume that for any $\Phi \in C_c ((0,\infty) )$
\begin{gather*}
\int \Phi(x)\rho_1^{(N)}(x){\rm d}x\to \int \Phi(x)\rho_1(x){\rm d}x.
\end{gather*}
Finally, assume that
\begin{gather*}
\lim_{\epsilon \to 0} \int_{0}^{\epsilon}x \rho_1^{(N)}(x){\rm d}x=0, \qquad \text{uniformly in $N$}.
\end{gather*}
Then, we have
\begin{gather*}
\gamma_1(X)=\sum_{i=1}^{\infty}\alpha_i^+(X), \qquad \text{for} \quad \mathfrak{M}-\operatorname{a.e.} X \in H_{{\rm reg}}.
\end{gather*}
\end{Proposition}

We first need an elementary lemma.
\begin{Lemma}Assume we are given numbers $\forall\, N \ge 1$
\begin{gather*}
\alpha_{1,N}^+\ge \alpha_{2,N}^+\ge \cdots \ge 0
\end{gather*}
such that
\begin{gather*}
\lim_{N \to \infty}\alpha_{i,N}^+=\alpha_i^+, \qquad \forall\, i\ge 1.
\end{gather*}
Moreover, assume the following limit exists and is finite
\begin{gather*}
\lim_{N\to \infty} \sum_{i=1}^{\infty}\alpha_{i,N}^+=\gamma_1<\infty.
\end{gather*}
Note that by Fatou's lemma $($and positivity$)$
\begin{gather*}
\Delta=\gamma_1-\sum_{i=1}^{\infty}\alpha_i^+\ge 0.
\end{gather*}
Let $\Phi$ be a continuous function on $(0,\infty)$ such that
\begin{gather*}
\Phi(x)=x, \qquad x<\epsilon
\end{gather*}
for a certain $\epsilon>0$. Then,
\begin{gather*}
\lim_{N \to \infty}\sum_{i=0}^{\infty} \Phi\big(\alpha_{i,N}^+\big)=\sum_{i=1}^{\infty} \Phi\big(\alpha_i^+\big)+\Delta.
\end{gather*}
\end{Lemma}

\begin{proof}Observe that, there exists $k$ such that $\alpha_{k+1}^{+}<\epsilon$. Then $\alpha_{k+1,N}^+<\epsilon$ for $N$ sufficiently large and $\alpha_{i,N}^+<\epsilon$ for $i\ge k+1$ by monotonicity. Also, $\alpha_i^+<\epsilon$ for $i\ge k+1$. Therefore
\begin{gather*}
\Phi(\alpha_{i,N}^+)=\alpha_{i,N}^+, \quad N \ \text{large},\qquad \Phi(\alpha_{i}^+)=\alpha_{i}^+, \quad \forall\, i \ge k+1.
\end{gather*}
Thus,
\begin{gather*}
\sum_{i=1}^{\infty}\Phi\big(\alpha_{i,N}^+\big)=\sum_{i=1}^{k}\Phi\big(\alpha_{i,N}^+\big)+\sum_{i=k+1}^{\infty}\alpha_{i,N}^+
\end{gather*}
and
\begin{gather*}
\sum_{i=1}^{\infty}\Phi\big(\alpha_{i}^+\big)=\sum_{i=1}^{k}\Phi\big(\alpha_{i}^+\big)+\sum_{i=k+1}^{\infty}\alpha_{i}^+.
\end{gather*}
As $N \to \infty$ by continuity of $\Phi$
\begin{gather*}
\sum_{i=1}^{k}\Phi\big(\alpha_{i,N}^+\big) \to \sum_{i=1}^{k}\Phi\big(\alpha_{i}^+\big)
\end{gather*}
and by the assumptions of the lemma
\begin{gather*}
\sum_{i=k+1}^{\infty}\alpha_{i,N}^+\to \sum_{i=k+1}^{\infty}\alpha_{i}^++\Delta.
\end{gather*}
The statement now follows.
\end{proof}

\begin{proof}[Proof of Proposition \ref{AbstractGamma1}]
First, observe that $\mathfrak{M}$ is supported on the subset $H^{*}_{{\rm reg}}\subset H_{{\rm reg}}$ that we now define. An element $X \in H^*_{{\rm reg}}$ iff
\begin{gather*}
\alpha_{i,N}^-(X)\equiv 0, \quad \forall\, i\ge 1, \quad N\ge 1, \qquad \alpha_i^-(X) \equiv 0, \quad \forall\, i \ge 1, \qquad \gamma_1(X)<\infty.
\end{gather*}
Fix a continuous function $\Phi(x)\ge 0$, vanishing for $x$ large enough, such that $\Phi(x)=x$ near 0. For any $X \in H^{*}_{{\rm reg}}$ we set
\begin{gather*}
\phi_N(X) =\sum_{i=1}^{\infty}\Phi\big(\alpha_{i,N}^+(X)\big),\qquad
\phi_{\infty}(X) =\sum_{i=1}^{\infty}\Phi\big(\alpha_i^+(X)\big).
\end{gather*}
Apply the previous lemma to the sequences $\alpha_{i,N}^+=\alpha_{i,N}^+(X)$, $\alpha_{i}^+=\alpha_{i}^+(X)$ for $X \in H^{*}_{{\rm reg}}$ (note that all conditions are satisfied) to get
\begin{gather*}
\phi_N(X)\to \phi_{\infty}(X)+\Delta(X).
\end{gather*}
Observe that all three functions $\phi_N$, $\phi_{\infty}$, $\Delta$ are non-negative and thus Fatou's lemma gives
\begin{gather*}
\liminf_{N\to \infty} \int_{X \in H^{*}_{{\rm reg}}}^{}\phi_N(X)\mathfrak{M}({\rm d}X) \ge \int_{X \in H^{*}_{{\rm reg}}}^{}\phi_\infty(X)\mathfrak{M}({\rm d}X)+\int_{X \in H^{*}_{{\rm reg}}}^{}\Delta(X)\mathfrak{M}({\rm d}X).
\end{gather*}
Associate the point configurations $\mathsf{C}_N(X)$, $\mathsf{C}(X)$ to $X\in H^{*}_{{\rm reg}}$. Then
\begin{gather*}
\phi_{N}(X)=\sum_{i=1}^{\infty} \Phi\big(\alpha_{i,N}^+(X)\big)=\sum_{x \in \mathsf{C}_N(X)}^{}\Phi(x)
\end{gather*}
so that by the definition of the correlation functions
\begin{gather*}
\int_{X \in H^{*}_{{\rm reg}}}^{}\phi_N(X)\mathfrak{M}({\rm d}X)=\int\Phi(x)\rho_1^{(N)}(x){\rm d}x
\end{gather*}
and similarly
\begin{gather*}
\int_{X \in H^{*}_{{\rm reg}}}^{}\phi_\infty(X)\mathfrak{M}({\rm d}X)= \int \Phi(x)\rho_1(x){\rm d}x.
\end{gather*}
Thus
\begin{gather*}
\liminf_{N \to \infty}\int\Phi(x)\rho_1^{(N)}(x){\rm d}x \ge \int \Phi(x)\rho_1(x){\rm d}x+\int_{X \in H^{*}_{{\rm reg}}}^{}\Delta(X)\mathfrak{M}({\rm d}X).
\end{gather*}
We now proceed to show that
\begin{gather*}
\limsup_{N \to \infty}\int\Phi(x)\rho_1^{(N)}(x){\rm d}x \le \int \Phi(x)\rho_1(x){\rm d}x.
\end{gather*}
Since $\Delta(X)\ge 0$ on $H^{*}_{{\rm reg}}$ we get that $\Delta(X)\equiv 0$ for $\mathfrak{M}-\operatorname{a.e.} X \in H^*_{{\rm reg}}$, from which, recalling that $\mathfrak{M}$ is supported on $H^*_{{\rm reg}}$ the conclusion of the proposition follows.

To this end, decompose $\Phi(x)$ as follows, for arbitrary $\epsilon>0$
\begin{gather*}
\Phi(x)=\Phi_{\epsilon}(x)+\Psi_{\epsilon}(x),
\end{gather*}
where $0\le \Phi_{\epsilon}(x)\le x$, $\operatorname{supp} \Phi_{\epsilon} \subset [0,\epsilon]$, $\Phi_{\epsilon}(x)=x$ near $0$ and $\Psi_{\epsilon}\in C_{c} ((0,\infty) )$ and positive. From the assumption of the proposition we obtain
\begin{gather*}
\lim_{\epsilon \to 0} \limsup_{N \to \infty}\int\Phi_{\epsilon}(x)\rho_1^{(N)}(x){\rm d}x=0.
\end{gather*}
Thus by Fatou's lemma for any $\epsilon>0$
\begin{gather*}
\limsup_{N \to \infty}\int\Phi(x)\rho_1^{(N)}(x){\rm d}x\le \limsup_{N \to \infty}\int\Phi_{\epsilon}(x)\rho_1^{(N)}(x){\rm d}x+\limsup_{N \to \infty}\int\Psi_{\epsilon}(x)\rho_1^{(N)}(x){\rm d}x.
\end{gather*}
Taking the limit $\epsilon \to 0$ we finally get, by convergence of the first correlation function $\rho_1^{(N)} \to \rho_1$,
\begin{align*}
\limsup_{N \to \infty}\int\Phi(x)\rho_1^{(N)}(x){\rm d}x&\le \lim_{\epsilon \to 0}\limsup_{N \to \infty}\int\Phi_{\epsilon}(x)\rho_1^{(N)}(x){\rm d}x+ \lim_{\epsilon \to 0}\limsup_{N \to \infty}\int\Psi_{\epsilon}(x)\rho_1^{(N)}(x){\rm d}x\\
&=\lim_{\epsilon \to 0}\int\Psi_{\epsilon}(x)\rho_1(x){\rm d}x=\int \Phi(x)\rho_1(x){\rm d}x.\tag*{\qed}
\end{align*}\renewcommand{\qed}{}
\end{proof}

\begin{Proposition}\label{FinitenessGamma1}Let $\nu>-1$. Then
\begin{gather*}
\gamma_1(X)<\infty, \qquad \text{for} \quad \mathsf{M}^{(\nu)}-\operatorname{a.e.} X \in H_{{\rm reg}}.
\end{gather*}
\end{Proposition}

\begin{proof}First of all, we note that $\mathsf{M}^{(\nu)}$ is supported on the subset $H^{+}_{{\rm reg}}\subset H_{{\rm reg}}$ that we now define. An element $X \in H^+_{{\rm reg}}$ iff
\begin{gather*}
\alpha_{i,N}^-(X)\equiv 0, \quad \forall\, i\ge 1, \quad N\ge 1, \qquad \alpha_i^-(X) \equiv 0, \quad \forall\, i \ge 1.
\end{gather*}
Moreover, if we define for $R>0$ the subset $H^{+,R}_{{\rm reg}}\subset H^{+}_{{\rm reg}}$ such that $X \in H^{+,R}_{{\rm reg}}$ iff $\alpha_1^+(X)<R$ we easily see that
\begin{gather*}
H^{+}_{{\rm reg}}=\bigcup_{k \in \mathbb{N}} H^{+,k}_{{\rm reg}}.
\end{gather*}
Hence it will suffice to show that for any fixed $R>0$
\begin{gather*}
\gamma_1(X)<\infty, \qquad \text{for} \quad \mathsf{M}^{(\nu)}-\operatorname{a.e.} X \in H^{+,R}_{{\rm reg}}.
\end{gather*}
Furthermore, by positivity it actually suffices to show
\begin{gather*}
\mathbb{E}\big[\gamma_1(X) \textbf{1}\big(X \in H_{{\rm reg}}^{+,R}\big)\big]<\infty,
\end{gather*}
where the expectation $\mathbb{E}$ is with respect to $\mathsf{M}^{(\nu)}$. We calculate, using Fatou's lemma and the underlying determinantal structure
\begin{align*}
\mathbb{E}\big[\gamma_1(X) \textbf{1}\big(X \in H_{{\rm reg}}^{+,R}\big)\big]&=\mathbb{E}\big[\gamma_1(X) \textbf{1}\big(\alpha_1^+(X)<R\big)\big]\\
&=\mathbb{E}\left[\lim_{N\to \infty} \left(\textbf{1}\big(\alpha_1^+(X)<R\big)\sum_{i=1}^{\infty}\alpha_{i,N}^+(X)\right)\right]\\
&=\mathbb{E}\left[\lim_{N\to \infty} \left(\textbf{1}\big(\alpha_{1,N}^+(X)<R\big)\sum_{i=1}^{\infty}\alpha_{i,N}^+(X)\right)\right]\\
&\le\liminf_{N \to \infty}\mathbb{E}\left[\textbf{1}\big(\alpha_{1,N}^+(X)<R\big)\sum_{i=1}^{\infty}\alpha_{i,N}^+(X)\right]\\
&=\liminf_{N \to \infty}\mathbb{E}\left[\sum_{x \in \mathsf{C}_N^{(\nu)}(X)}x\textbf{1} (x<R )\right]\\
&=\liminf_{N \to \infty}\int_{0}^{R}x K_N^{\nu}(x,x){\rm d}x<\infty.
\end{align*}
The last claim is the statement of Lemma~\ref{LemmaUniformBoundedness}.
\end{proof}

\begin{proof}[Proof of Proposition \ref{PropositionGamma1}] We apply Proposition \ref{AbstractGamma1}. The first assumptions follow from Pro\-po\-si\-tions~\ref{PropositionAlphaMinus},~\ref{PropositionAlphaPlus} and~\ref{FinitenessGamma1} above, while the fact that
\begin{gather*}
\lim_{\epsilon \to 0} \int_{0}^{\epsilon}x \rho_1^{(N)}(x)\,{\rm d}x=0, \qquad \text{uniformly in $N$}.
\end{gather*}
follows from Proposition \ref{EstimateGamma1}.
\end{proof}
\section{Proof of main theorem}\label{section8}

\begin{proof}[Proof of Theorem \ref{MainTheorem}]
The fact that
\begin{gather*}
\mathfrak{m}^{(\nu)}\big(\Omega_0^+\big)=1
\end{gather*}
follows from combining Propositions~\ref{PropositionAlphaMinus},~\ref{PropositionAlphaPlus},~\ref{PropositionGamma2}, and~\ref{PropositionGamma1}.

The description of the law of the parameters $\alpha^+=\left(\alpha^+_1\ge \alpha_2^+\ge\alpha_3^+\ge \cdots \ge 0\right)$ under $\mathfrak{m}^{(\nu)}$ viewed as a point configuration on $(0,\infty)$ is given by Proposition \ref{PropositionLimitCorrelationAlphaPlus} (after making use of Theorem~\ref{ErgodicMethodTheorem}). The proof is complete.
\end{proof}

\subsection*{Acknowledgements}
I would like to thank Alexei Borodin and Grigori Olshanski for some useful comments and pointers to the literature. Finally, I would like to thank the anonymous referees for a careful reading of the paper and a number of useful suggestions and remarks. Research supported by ERC Advanced Grant 740900 (LogCorRM).

\pdfbookmark[1]{References}{ref}
\LastPageEnding

\end{document}